\documentclass{amsart}

\usepackage{amsxtra}
\usepackage{amssymb, amscd, amsfonts}
\usepackage{graphicx}
\usepackage{setspace}
\usepackage{amsthm}

\newcommand\lw[1]{\hbox{}_{#1}\!}
\newcommand\CC{{\mathbb C}}
\newcommand\SU{{\mathrm{SU}}}
\newcommand\U{{\mathrm{U}}}
\newcommand\ZZ{{\mathbb Z}}
\newcommand\NN{{\mathbb N}}

\theoremstyle{plain}
\newtheorem{thm}{Theorem}[section]

\theoremstyle{definition}
\newtheorem{defn}[thm]{Definition}

\theoremstyle{remark}
\newtheorem{remark}[thm]{Remark}

\begin{document}

\title{The generalized matrix valued hypergeometric equation}

\author{P. Rom\'an}
\address{CIEM-FaMAF, Universidad Nacional de C\'or\-do\-ba, Medina Allende s/n
Cidudad Universitaria, C\'or\-do\-ba~5000,
Argentina, FAX:+54 351 4334054}
\email{roman@mate.uncor.edu}

\author{S. Simondi}
\address{CIEM-FaMAF, Universidad Nacional de C\'or\-do\-ba, Medina Allende s/n
Cidudad Universitaria, C\'or\-do\-ba~5000, Argentina, FAX:+54 351 4334054}
\email{simondi@mate.uncor.edu}

\begin{abstract}
The matrix valued analog of the Euler's hypergeometric differential equation
was introduced by Tirao in \cite{T2}. This equation arises in the study of
matrix valued spherical functions
and in the theory of matrix valued orthogonal polynomials. The goal of this
paper is to extend naturally the number of parameters of Tirao's equation in
order to get a generalized matrix valued hypergeometric equation. We take
advantage of the tools and strategies developed in \cite{T2} to identify the
corresponding matrix hypergeometric functions ${}_nF_m$. We prove that, if
$n=m+1$, this functions are analytic for $|z|<1$ and we give a necesary
condition for the convergence on the unit circle $|z|=1$.
\end{abstract}



\maketitle

\bibliographystyle{unsrt}

\section{Introduction}
The importance of the hypergeometric differential equation introduced by
Euler in 1769 and the hypergeometric function was perceived by famous
mathematicians as Gauss, Kummer and Riemman. Since then, many others found generalizations and
applications. In fact many of the special functions that appear in mathematical
physics, engineering and statistics are special cases of hypergeometric
functions. Every second-order ordinary differential equation with three regular
singular
points, by placing the singularities in at $0$, $1$, and $\infty$, can be
transformed into the hypergeometric differential equation
$$z(1-z)f''(z)+(c-z(a+b+1))f'(z)-abf(z)=0,$$
where $a$, $b$ and $c$ are supposed to be complex numbers. If $c$ is not an
integer we can verify by a direct differentiation of the series that the unique
solution of this equation, analytic and with value $1$ at $z=0$, is given by
the hypergeometric function defined by
$$\lw{2}F_1 \left( \begin{smallmatrix} a \, , \, b \\ c \end{smallmatrix}
;z\right)=\sum_{n\geq 0} \frac{(a)_n (b)_n}{(c)_n}\frac{z^n}{n!}, \text{ where } (w)_n=w(w+1)\ldots(w+n-1).$$
Recently, Tirao introduced in \cite{T2} the matrix valued analog of the
hypergeometric differential equation
\begin{align}
\label{hipergeometrica_Tirao}
z(1-z)F''(z)+(C-z(A+B+1))F'(z)-ABF(z)=0,
\end{align}
where $A, B, C\in \CC^{r\times r}$ and $F$ denotes a
function on $\CC$ with values in $\CC^r$. The corresponding matrix analog of the
Gauss' hypergeometric function is given by
$$\lw{2}F_1 \left( \begin{smallmatrix} A \, ; \, B \\ C \end{smallmatrix}
;z\right)=\sum_{n\geq 0} \left(\begin{smallmatrix} A \, ; \, B \\ C
\end{smallmatrix}\right)_n \frac{z^n}{n!},$$
where $\left(\begin{smallmatrix} A \, ; \, B \\ C
\end{smallmatrix}\right)_0$ is the identity matrix and
$$\left(\begin{smallmatrix} A \, ; \, B \\ C
\end{smallmatrix}\right)_{n+1}=(C+n)^{-1}(A+n)(B+n)\left(\begin{smallmatrix} A
\, ; \, B \\ C \end{smallmatrix}\right)_{n},\text{ for all }n\geq 0.$$
The matrix valued hypergeometric functions have connections with
the theory of matrix valued spherical functions and with the theory of matrix
valued orthogonal polynomials (\cite{RT},\cite{PT1},\cite{PR}). In fact in \cite{RT} we prove
that the matrix representations of irreducible spherical functions associated to
the dual symmetric pairs $(\SU(2,1),\U(2))$ and $(\SU(3),\U(2))$ are given in terms of matrix
hypergeometric functions.

The goal of this paper is to generalize the matrix valued hypergeometric 
equation \eqref{hipergeometrica_Tirao} by extending naturally the number of
parameters as it was done in the scalar case at the end of the nineteenth century. In
this way we shall consider the following generalized matrix hypergeometric
equation
\begin{multline}
\label{hyperg_generalizada}
z\frac{\operatorname{d}}{\operatorname{d}z}\left(z\frac{\operatorname{d}}{\operatorname{d}z}+B_1-1\right)
\left(z\frac{\operatorname{d}}{\operatorname{d}z}+B_2-1\right)\ldots\left(z\frac{\operatorname{d}}
{\operatorname{d}z}+B_m-1\right)F(z)\\
-z\left(z\frac{\operatorname{d}}{\operatorname{d}z}+A_1\right)\left(z\frac{\operatorname{d}}{\operatorname{d}z}+A_2\right)
\ldots\left(z\frac{\operatorname{d}}{\operatorname{d}z}+A_n\right)F(z)=0,
\end{multline}
where $n,m\in \ZZ_{\geq0}$, $A_1,\ldots,A_n,B_1,\ldots,B_m \in \CC^{r\times r}$
and $F$ denotes a complex function with values in $\CC^r$. Observe that if we set $n=2$
and $m=1$ we obtain the Tirao's hypergeometric equation.

If no eigenvalue of $B_1,\ldots,B_m$ is in the set $\{0,-1,-2,\ldots\}$ we introduce the
generalized matrix valued hypergeometric function in the following way
\begin{equation*}
{}_n F_m\left( \begin{smallmatrix} A_1;\,\,\ldots\,\, ;A_n\\ B_1;\,\ldots \,;B_m
\end{smallmatrix};z\right)=\sum_{j=0}^\infty \frac{z^j}{j!} \left( \begin{smallmatrix} A_1;\,\,\ldots\,\, ;A_n\\ B_1;\,\ldots \,;B_m
\end{smallmatrix}\right)_j,
\end{equation*}
where the symbol $\left( \begin{smallmatrix} A_1;\,\,\ldots\,\, ;A_n\\ B_1;\,\ldots \,;B_m
\end{smallmatrix}\right)_j$ is defined in \eqref{simbologeneralizado}. We prove
that this functions are solutions of the differential equation
\eqref{hyperg_generalizada}. Furthermore if $n=m+1$
and $\{V_j\}_{j=1}^r$ is a basis of $\CC^r$ then the set
$$\left\{{}_{m+1} F_m\left( \begin{smallmatrix}
A_1;\,\,\ldots\,\, ;A_{m+1}\\ B_1;\,\ldots \,;B_m \end{smallmatrix};z\right) V_j\right\}_{j=1}^r,$$
is basis of the space of all solutions
of the hypergeometric equation \eqref{hyperg_generalizada} analytic at $z=0$ for
$|z|<1$.

Throughout this paper for any matrix $A\in \CC^{r\times r}$ we denote by
$\sigma(A)$ to the set of all eigenvalues of $A$ and
$$\rho(A)=max\{Re(\lambda):\lambda\in \sigma(A)\},\,\,\,\,\, \delta(A)=min\{Re(\lambda):\lambda\in
\sigma(A)\}.$$
For any $A\in \CC^{r\times r}$, we shall consider the spectral
norm $\|\bullet\|$ defined by
$$\|A\|=\max \{\sqrt{\lambda}:\lambda\text{ is an eigenvalue of }AA^*\}.$$
Let us assume that $(\sigma(B_1)\cup\sigma(B_2)\cup\ldots\cup
\sigma(B_m))\cap(-\NN_0)=\emptyset$ and that $B_1,\ldots,B_{m}$ are unitarily
equivalent to a diagonalizable matrix with real eigenvalues. If
$$\delta(B_1)+\ldots+\delta(B_{m})\geq \|A_1\|+\ldots+\|A_{m+1}\|,$$
in Theorem \ref{convergencia_en_uno} we prove that ${}_{m+1} F_{m}\left(
\begin{smallmatrix} A_1;\,\,\ldots\,\, ;A_{m+1}\\ B_1;\,\ldots \,;B_{m}
\end{smallmatrix};z\right)$ is absolutely convergent on the unit circle $|z|=1$.

\section{On the generalized matrix valued hypergeometric equation}

The aim of this section is to give a generalization of the matrix valued
hypergeometric equation to the case of an arbitrary number of parameters. Let
$A_1,\ldots,A_n,B_1,\ldots,B_m\in\CC^{r\times r}$. We begin by considering the
following matrix valued differential equation of degree $\max(n,m)$.
\begin{multline}
\label{hyperg_generalizada2}
z\frac{\operatorname{d}}{\operatorname{d}z}\left(z\frac{\operatorname{d}}{\operatorname{d}z}+B_1-1\right)
\left(z\frac{\operatorname{d}}{\operatorname{d}z}+B_2-1\right)\ldots\left(z\frac{\operatorname{d}}
{\operatorname{d}z}+B_m-1\right)F(z)\\
-z\left(z\frac{\operatorname{d}}{\operatorname{d}z}+A_1\right)\left(z\frac{\operatorname{d}}{\operatorname{d}z}+A_2\right)
\ldots\left(z\frac{\operatorname{d}}{\operatorname{d}z}+A_n\right)F(z)=0.
\end{multline}
Let us denote by $\Delta$ to the differential operator
$z\frac{\operatorname{d}}{\operatorname{d}z}$. Then, in terms of
such operator, the matrix valued differential equation
\eqref{hyperg_generalizada2} is given by
\begin{multline}
\label{ecuacioncondeltas}
\Delta(\Delta+B_1-1)(\Delta+B_2-1)\ldots(\Delta+B_m-1)F(z)\\
=z(\Delta+A_1)(\Delta+A_2)\ldots(\Delta+A_n)F(z).
\end{multline}
We will seek for solutions which have the following form
$$F(z)=z^p\sum_{j=0}^{\infty}F_jz^j,$$
where $F_j\in \CC^r$. First at all we observe that
\begin{align*}
\Delta F=z\frac{\operatorname{d}}{\operatorname{d}z}\sum_{j=0}^\infty F_j
z^{p+j}=\sum_{j=0}^\infty(p+j)F_jz^{p+j},
\end{align*}
and therefore we have
\begin{multline*}(\Delta+A_i) F=\sum_{j=0}^\infty(A_i+p+j)F_jz^{p+j},\quad 
(\Delta+B_k-1) F=\sum_{j=0}^\infty(B_k+p+j-1)F_jz^{p+j},
\end{multline*}
for all $1\leq i \leq n$ and $0\leq k\leq m$. If we substitute the
previous results in the differential equation
\eqref{ecuacioncondeltas} we obtain the following recursion relation
for the $\CC^{r}$-valued coefficients $F_j$ of $F(z)$
\begin{multline} \label{200}
(p+j+1)(B_1+p+j)\ldots (B_m+p+j)F_{j+1} \\
=(A_1+p+j)\ldots(A_n+p+j)F_j.
\end{multline}
If we set $j=-1$ in the recursion relation we obtain
\begin{equation}
\label{indicial}
p(B_1+p-1)(B_2+p-1)\ldots(B_m+p-1)F_0=0.
\end{equation}
Therefore we have the following indicial equation
\begin{align}
\label{indicial_determinante}
p^r\det(B_1+p-1)\det(B_2+p-1)\ldots\det(B_m+p-1)=0.
\end{align}
Let $\beta^j_1,\ldots,\beta^j_r$ be the eigenvalues of $B_j$. Then the roots of
the indicial equation are given by
$$p=0,1-\beta^1_1,\ldots,1-\beta^1_r,1-\beta^2_1,\ldots,1-\beta^2_r,\ldots,1-\beta^m_1,\ldots,1-\beta^m_r.$$
The case $p=0$ correspond to the analytic solutions of the
differential equation \eqref{hyperg_generalizada2}.

First at all we will describe the analytc solutions of the equation
\eqref{hyperg_generalizada2}. Let $\left( \begin{smallmatrix} A_1;\,\,\ldots\,\, ;A_n\\ B_1;\,\ldots \,;B_m
\end{smallmatrix}\right)_j$ be the symbol defined recursively by 
\begin{equation}
\label{simbologeneralizado}
\begin{split}
\left( \begin{smallmatrix} A_1;\,\,\ldots\,\, ;A_n\\
B_1;\,\ldots \,;B_m \end{smallmatrix}\right)_0 &= I,\\
\left( \begin{smallmatrix} A_1;\,\,\ldots\,\, ;A_n\\ B_1;\,\ldots
\,;B_m
\end{smallmatrix}\right)_{j+1}&=
(B_m+j)^{-1}\ldots(B_1+j)^{-1}(A_1+j)\ldots(A_n+j)\left(
\begin{smallmatrix} A_1;\,\,\ldots\,\, ;A_n\\ B_1;\,\ldots \,;B_m
\end{smallmatrix} ; z\right)_j.
\end{split}
\end{equation}
\begin{defn}
If $A_1,\ldots,A_n,B_1,\ldots,B_m \in \CC^{r\times r}$ and no
eigenvalue of $B_1,\ldots,B_m$ is in the set
$\{0,-1,-2,\ldots\}$, then we define the funtion
\begin{equation}
\label{hypergeom_generalizada}
{}_n F_m\left( \begin{smallmatrix} A_1;\,\,\ldots\,\, ;A_n\\ B_1;\,\ldots \,;B_m
\end{smallmatrix};z\right)=\sum_{j=0}^\infty \frac{z^j}{j!} \left( \begin{smallmatrix} A_1;\,\,\ldots\,\, ;A_n\\ B_1;\,\ldots \,;B_m
\end{smallmatrix}\right)_j.
\end{equation}
\end{defn}

\begin{thm}
The function $\lw{n}F_m \left( \begin{smallmatrix} A_1;\,\,\ldots\,\, ;A_n\\ B_1;\,\ldots
\,;B_m \end{smallmatrix} ; z\right)$ is analytic for all $z\in \CC$ if $n\leq m$
and it is analytic on $|z|<1$ if $n=m+1$. On the other hand if $n>m+1$ the
series is not absolutely convergent unless $|z|=0$.
\end{thm}
\begin{proof}
Let $\|\cdot\|$ be any matrix norm on $\CC^{r\times r}$. We observe that
\begin{align*}
\left\|\lw{n}F_m \left( \begin{smallmatrix} A_1;\,\,\ldots\,\, ;A_n\\ B_1;\,\ldots
\,;B_m \end{smallmatrix} ; z\right)\right\|&=\left\| \sum_{j=0}^\infty \left( \begin{smallmatrix} A_1;\,\,\ldots\,\, ;A_n\\
B_1;\ldots ;B_m \end{smallmatrix}\right)_j \frac{z^j}{j!}\right\| \leq
\sum_{j=0}^\infty \left\|\left( \begin{smallmatrix} A_1;\,\,\ldots\,\, ;A_n\\ B_1;\,\ldots
\,;B_m \end{smallmatrix}\right)_j \right\| \frac{|z|^j}{j!}.
\end{align*}
Now we will prove that the power series $\sum_{j=0}^\infty \left\|\left(
\begin{smallmatrix} A_1;\,\,\ldots\,\, ;A_n\\ B_1;\,\ldots
\,;B_m \end{smallmatrix}\right)_j \right\| \frac{|z|^j}{j!}$ is convergent. 
In this way we denote $a_j=\left\|\left(
\begin{smallmatrix} A_1;\,\,\ldots\,\, ;A_n\\ B_1;\,\ldots
\,;B_m \end{smallmatrix}\right)_j \right\|\frac{|z|^j}{j!}$. Then we shall
compute $\lim_{j\rightarrow \infty} \frac{a_{j+1}}{a_j}$. Observe that
\begin{align*}
\frac{a_{j+1}}{a_j}&=\frac{|z|}{j+1}\left\|\left( \begin{smallmatrix} A_1;\,\,\ldots\,\, ;A_n\\ B_1;\,\ldots
\,;B_m \end{smallmatrix}\right)_{j+1}\right\|\left\|\left( \begin{smallmatrix} A_1;\,\,\ldots\,\, ;A_n\\ B_1;\,\ldots
\,;B_m \end{smallmatrix}\right)_j \right\|^{-1}\\
&\leq \frac{|z|}{j+1}\left\|
(B_m+j)^{-1}\ldots(B_1+j)^{-1}(A_1+j)\ldots(A_n+j)\right\|\\
&=\frac{|z|j^{n-m-1}}{\frac{1}{j}+1} \left\|\textstyle{
\left(1+\frac{B_m}{j}\right)^{-1}\ldots\left(1+\frac{B_1}{j}\right)^{-1}\left(1+\frac{A_1}{j}\right)\ldots\left(1+\frac{A_n}{j}\right)} \right\|,
\end{align*}
Provided that $n\leq m$, this expression tends to zero as
$j\rightarrow \infty$ and thus the series is not absolutely convergent
unless $z=0$. If $n>m+1$, the series is convergent for all $z$. On the
other side, for $n=m+1$ the series converges for $|z|<1$.
\end{proof}

It is worth noticing that if $F_0\in \CC^{r}$ then ${}_n F_m\left( \begin{smallmatrix} A_1;\,\,\ldots\,\, ;A_n\\ B_1;\,\ldots \,;B_m
\end{smallmatrix};z\right)F_0$ is a solution of the differential equation
\eqref{hyperg_generalizada2}. In the next theorem we summarize our
results which characterize the analytic solutions at $z=0$ of the
generalized hypergeometric equation
\begin{thm}\label{analiticas}
Let us assume that $(\sigma(B_1)\cup\sigma(B_2)\cup\ldots\cup
\sigma(B_m))\cap(-\NN_0)=\emptyset$ and let $F_0\in \CC^r$ then $F(z)={}_n F_m\left( \begin{smallmatrix} A_1;\,\,\ldots\,\, ;A_n\\ B_1;\,\ldots \,;B_m
\end{smallmatrix};z\right)F_0$ is a solution of the differential equation \eqref{hyperg_generalizada2}
such that $F(0)=F_0$. Conversely any solution $F$ analytic at $z=0$ is of this form.
\end{thm}
For any $A\in \CC^{r\times r}$, we shall consider the spectral
norm $\|\bullet\|$ defined by
$$\|A\|=\max \{\sqrt{\lambda}:\lambda\text{ is an eigenvalue of }AA^*\}.$$
Since the spectral norm
is a unitarily invariant matrix norm we observe that $\|XAX^{*}\|=\|A\|$
for any
$A\in \CC^{r\times r}$ and any unitary matrix $X\in \CC^{r\times r}$.
\begin{thm}
\label{convergencia_en_uno}
Let $A_1,\ldots,A_{m+1},B_1,\ldots,B_{m}\in\CC^{r\times r}$ such that $(\sigma(B_1)\cup\sigma(B_2)\cup\ldots\cup
\sigma(B_m))\cap(-\NN_0)$ is empty and let $B_1,\ldots,B_{m}$ be
unitarily equivalent to a diagonalizable matrix with real
eigenvalues. If
$$\delta(B_1)+\ldots+\delta(B_{m})\geq \|A_1\|+\ldots+\|A_{m+1}\|,$$
then the matrix hypergeometric function ${}_{m+1} F_{m}\left(
\begin{smallmatrix} A_1;\,\,\ldots\,\, ;A_{m+1}\\ B_1;\,\ldots \,;B_{m}
\end{smallmatrix};z\right)$ is absolutely convergent for $|z|=1$.
\end{thm}
\begin{proof}
Observe that since $B_i$ is diagonalizable, there exists $X_i\in
\CC^{r\times r}$ such that $X_i B_i X_i^{*}=\Lambda_i$ where
$\Lambda_i$ is a diagonal matrix with real no zero entries. Then
we have that
$$\|B_i^{-1}\|=\|X_i B_i^{-1} X_i^{*}\|=\|\Lambda_i\|=\rho(B_i^{-1})=\frac{1}{\delta(B_i)},$$
and
$$\|(B_i+kI)^{-1}\|=\frac{1}{\delta(B_i+kI)}=\frac{1}{\delta(B_i)+k}.$$
By hypothesis, there exists a positive number $\lambda$ such that
$$\sum_{i=1}^{m} \delta(B_i)-\sum_{i=1}^{m+1}\|A_i\|=2\lambda,$$ 
and we have that
\begin{align*}
\left\| \frac{|z|^{k}}{k!}  \left(
\begin{smallmatrix} A_1;\,\,\ldots\,\, ;A_{m+1}\\ B_1;\,\ldots \,;B_{m}
\end{smallmatrix}\right)_k  \right\|  &\leq \frac{|z|^{k}}{k!}
\prod_{i=1}^{m}\left(\prod_{h=0}^{k-1}\|(B_i+hI)^{-1}\|\right) \, \prod_{j=1}^{m+1} (\|A_j\|)_k
\\
&= \frac{|z|^{k}}{k!} \prod_{i=1}^{m} \frac{1}{(\delta(B_i))_k} \,
\prod_{j=1}^{m+1} (\|A_j\|)_k \\
&\leq |z|^{k} \, k^{\sum_{i=1}^{m} \delta(B_i)-\sum_{i=1}^{m+1}\|A_i\|-1} \,
\prod_{i=1}^{m}\left( \frac{(k-1)!\,
k^{\delta(B_i)}}{(\delta(B_i))_k} \right) \\
& \hspace{3.9cm} \times \prod_{i=1}^{m+1}\left(
\frac{(\|A_i\|)_k}{(k-1)!\, k^{\|A_i\|}} \right).
\end{align*}
If we set $|z|=1$, then
\begin{align*}
\lim_{k\rightarrow \infty} k^{1+\lambda} \left\| \frac{1}{k!}
\left(
\begin{smallmatrix} A_1;\,\,\ldots\,\, ;A_{m+1}\\ B_1;\,\ldots \,;B_{m}
\end{smallmatrix}\right)_k  \right\|& \leq  \lim_{k\rightarrow \infty} k^{-\lambda} \,
\prod_{i=1}^{m}\left( \frac{(k-1)!\,
k^{\delta(B_i)}}{(\delta(B_i))_k} \right) \\
& \hspace{1cm} \times\prod_{i=1}^{m+1}\left(
\frac{(\|A_i\|)_k}{(k-1)!\, k^{\|A_i\|}} \right) \displaybreak[0] \\ &= 0 \cdot
\prod_{i=1}^{m} \Gamma(\delta(B_i)) \prod_{j=1}^{m+1} \Gamma(\|A_j
\|)^{-1}=0.
\end{align*}
Since the series $\sum_{k=1}^{\infty}\frac{1}{k^{1+\lambda}}$
converge if $\lambda >0$, the Dirichlet's criterium of numerical
series of positive numbers implies the absolute convergence on the unit circle
for $|z|=1$ of the matrix hypergeometric function ${}_{m+1} F_{m}\left(
\begin{smallmatrix} A_1;\,\,\ldots\,\, ;A_{m+1}\\ B_1;\,\ldots \,;B_{m}
\end{smallmatrix};z\right)$.
\end{proof}
Now we will concern ourselves with the non analytic solutions of the
matrix equation \eqref{hyperg_generalizada2}. By taking into account the
recursion relation \eqref{200} we introduce the following definition.
\begin{defn} Let us
assume that $A_1,\ldots,A_n,B_1,\ldots,B_m$ are $r\times r$ complex matrices
and suppose that $-p\notin
(\sigma(B_1)+\NN_0)\cup\ldots\cup(\sigma(B_m)+\NN_0)\cup\NN$.
Then we define the function
\begin{equation}
\label{hypergeom_generalizada_conp} {}_n F_m^p\left(
\begin{smallmatrix} A_1;\,\,\ldots\,\, ;A_n\\ B_1;\,\ldots \,;B_m
\end{smallmatrix};z\right)=\sum_{j=0}^\infty \frac{z^j}{(p+1)_j} \left( \begin{smallmatrix} A_1+p;\,\,\ldots\,\, ;A_n+p\\ B_1+p;\,\ldots \,;B_m+p
\end{smallmatrix}\right)_j.
\end{equation}
Observe that $p=0$ implies that ${}_n F_m^p\left(
\begin{smallmatrix} A_1;\,\,\ldots\,\, ;A_n\\ B_1;\,\ldots \,;B_m
\end{smallmatrix};z\right)={}_n F_m\left( \begin{smallmatrix} A_1;\,\,\ldots\,\, ;A_n\\ B_1;\,\ldots \,;B_m
\end{smallmatrix};z\right)$.
\end{defn}
If $\beta \in
\sigma(B_1)\cup \dots \cup \sigma(B_m)$ and $\beta
\neq 1$ then $p=1-\beta$ is a nonzero solution of the indicial
equation 
\begin{equation*}
p^r\det(B_1+p-1)\det(B_2+p-1)\ldots\det(B_m+p-1)=0.
\end{equation*}
Thus the kernel of the matrix
$p(B_1+p-1)\ldots(B_m+p-1)$ is nonempty and we can take a nonzero vector
$F_{\beta} \in \ker(p(B_1+p-1)\ldots(B_m+p-1))$. Then, because of the way it was
constructed, the $\CC^{r}$-valued function 
$$F(z)=z^{1-\beta}{}_n
F_m^{1-\beta}\left( \begin{smallmatrix} A_1;\,\,\ldots\,\, ;A_n\\ B_1;\,\ldots \,;B_m 
\end{smallmatrix};z\right)F_{\beta},$$ 
is a solution of the matrix generalized hypergeometric equation \eqref{hyperg_generalizada2}. We ressume
this fact in the following theorem.
\begin{thm}\label{noanaliticas} Let $\beta \in
\sigma(B_1)\cup \dots \cup \sigma(B_m)$ and $\beta
\neq 1$. If $\beta=1-p$ and we assume that $\beta\not\in
(\sigma(B_1)+\NN)\cup\ldots\cup(\sigma(B_m)+\NN)\cup \NN$ and $F_{\beta}$ is a
vector in $\ker(p(B_1+p-1)\ldots(B_m+p-1))$ then
$$F(z)= z^{1-\beta}{}_n F_m^{1-\beta}\left( \begin{smallmatrix} A_1;\,\,\ldots\,\, ;A_n\\ B_1;\,\ldots \,;B_m
\end{smallmatrix};z\right)F_{\beta}=\sum_{j=0}^\infty \frac{z^{p+j}}{(p+1)_j}
\left( \begin{smallmatrix} A_1+p;\,\,\ldots\,\, ;A_n+p\\ B_1+p;\,\ldots
\,;B_m+p  \end{smallmatrix}\right)_j F_{\beta},$$
is a solution of the generalized hypergeometric equation \eqref{hyperg_generalizada2}.
\end{thm}

\begin{remark}
In the case that $n=m+1$, if the matrices $B_i$, $i=1,\ldots,m$, are all
diagonalizable, and the set of eigenvalues $\{\beta_i^j\}$ is such that
$\beta_i^j\neq \beta_k^h$ if $(i,j)\neq (k,h)$, the set of all analytic
solutions
of \eqref{hyperg_generalizada2} given in Theorem \ref{analiticas} and the set of
all non-analytic solutions given in Theorem \ref{noanaliticas} forms a
fundamental set of solutions of the hypergeometric equation
\eqref{hyperg_generalizada2}.
\end{remark}

\section{An example from the representation theory}

The importance of the study of the matrix valued hypergeometric equation, and its
solutions, become evident at the light of its connections with the theory of
spherical functions on any $K$-type on a Lie group and with the matrix valued
orthogonal polynomials (See \cite{RT,PT1,PR}). However the hypergeometric
equation that appear in \cite{RT} is slightly different from equation
\eqref{hipergeometrica_Tirao} and has the form
\begin{equation}
\label{hipergeometrica_esfericas2}
u(1-u)\frac{d^2F(u)}{du^2}+(C-uU)\frac{dF(u)}{du}-VF(u)=0,
\end{equation}
for $r\times r$ complex matrices $C, U$ and $V$ and a $\CC^r$-valued function
$F$ on $\CC$ . This differential equation cannot
always be reduced to a hypergeometric equation
\begin{align}
\label{hipergeometrica_tirao2}
z(1-z)F''(z)+(C-z(A+B+1))F'(z)-ABF(z)=0,
\end{align}
because we need to be able to find two matrices $A$ and $B$ which are solutions of the equations $U=A+B+1$, and $V=AB$,
or equivalently we need to solve the matrix quadratic equation $B^2+(1-U)B+V=0$.
This is not a trivial fact because a matrix quadratic equation may not have any solution. For
example there is not any $2\times2$ complex matrix $X$ such that 
$$X^2 -\left(\begin{matrix} 0&1\\ 0& 0 \end{matrix}\right)=0 .$$
Now we will concern ourselves in a particular case where it is posible to
transform the equation \eqref{hipergeometrica_esfericas2} into the
hypergeometric equation \eqref{hipergeometrica_tirao2}. Let
$Q(\lambda)=\lambda^2+\lambda(1-U)+V$. Then $\det(Q(\lambda))$ is a the
polynomial in $\lambda$ of degree $2r$ and therefore there are exactly $2r$ solutions
$\{\lambda_1,\ldots,\lambda_{2r}\}$ of the equation $\det(Q(\lambda))=0$. For
each $\lambda_i$, we have that the matrix $Q(\lambda_i)$ is singular. Let us
assume that we can choose $r$ different eigenvalues
$\lambda_{i_1},\ldots,\lambda_{i_r}$ and $r$ linearly independent vectors
$p_{i_1},\ldots,p_{i_r}$ such that $p_{i_k}\in \ker Q(\lambda_{i_k})$.
Let $\Lambda$ be the diagonal matrix whose $j$-th diagonal element is given by
$\lambda_{i_j}$ and let $X$ be the $r\times r$ matrix whose $j$-th
column is the vector $p_{i_j}$. Now we observe that
$$X\Lambda^2+(1-U)X\Lambda-VX\Lambda=0.$$
If we multiply the previous equation by $X^{-1}$ on the right, we obtain that
$$(X\Lambda X^{-1})^{2}+(1-U)X\Lambda X^{-1}-VX\Lambda X^{-1}=0,$$
and therefore $B=X\Lambda X^{-1}$ is a solution of the matrix quadratic equation
$B^2+(1-U)B+V=0$. We resume this discussion in the following theorem.
\begin{thm}
\label{encontrarAyB}
Let $C,U,V\in\CC^{r\times r}$ be such that we can choose $r$
different eigenvalues $\lambda_{i_1},\ldots,\lambda_{i_r}$ and $r$ linearly independent vectors
$p_{i_1},\ldots,p_{i_r}$ such that $p_{i_k}\in \ker Q(\lambda_{i_k})$. Let
$\Lambda=\operatorname{diag}(\lambda_{i_j})$ and let
$X=[p_{i_1},\ldots,p_{i_r}]$, i.e. that the $j$-th column of $X$ is $p_{i_j}$.
Then if we let
$$B=X\Lambda X^{-1},\quad A=U-X\Lambda X^{-1}-1,$$
we have that the differential equation \eqref{hipergeometrica_esfericas2} is the
following hypergeometric equation
$$z(1-z)F''(z)+(C-z(A+B+1))F'(z)-ABF(z)=0.$$
\end{thm}
\noindent {\it Example:} Let us consider the differential equation
\begin{equation}
\label{hipergeometrica_esfericas}
u(1-u)\frac{d^2F(u)}{du^2}+(C-uU)\frac{dF(u)}{du}-VF(u)=0,
\end{equation}
where the coefficient matrices are
\begin{align*}
C& =\sum_{i=0}^\ell (\beta+1+2i)E_{ii}+\sum_{i=1}^\ell iE_{i,i-1},
\quad
U=\sum_{i=0}^\ell (\alpha+\beta+\ell+i+2) E_{ii} , \displaybreak[0]\\
V&=  \sum_{i=0}^\ell i(\alpha+\beta+i-k+1)E_{ii}-
  \sum_{i=0}^{\ell-1} (\ell-i)(i+\beta-k+1)E_{i,i+1},
\end{align*}
with $\alpha,\beta >-1$, $0<k<\beta+1$ and  $\ell\in \NN$.

If we replace $\alpha$ by $m\in\ZZ_{\geq 0}$ and $\beta$ by $n-1$, this
differential equation arises from the first few steps in the
explicit determination of all matrix valued spherical functions associated to
the $n$-dimensional projective space $P_n(\CC)=\SU(n+1)/\U(n)$. Furthermore,
together whith an appropiate choice of a matrix weight $W(u)$, it provides
examples of families of Jacobi type matrix valued orthogonal polynomials (see
\cite{PT1} and \cite{PR}).

The goal is to use the content of Theorem \ref{encontrarAyB} to find matrices $A$ and $B$ 
such that $U=A+B+1$ and $V=AB$. First at all we observe that the matrix
$Q(\lambda)=\lambda^2+\lambda(1-U)+V$ is an upper triangular matrix whose
$i$-th diagonal entry is given by 
$$\lambda^2-\lambda(\alpha+\beta+\ell+i+1)+i(\alpha+\beta+i-k+1).$$
Therefore we have that $\det( Q(\lambda)) = \prod_{i=0}^\ell
(\lambda^2-\lambda(\alpha+\beta+\ell+i+1)+i(\alpha+\beta+i-k+1))$. Then we
can choose $\lambda_i$ as a solution of the cuadratic equation
$\lambda^2-\lambda(\alpha+\beta+\ell+i+1)+i(\alpha+\beta+i-k+1)=0$ and
a nonzero $v_i\in \ker Q(\lambda_i)$ for each $0\leq i\leq \ell$. For simplicity
we shall consider only the case in which $\lambda_i\neq\lambda_j$ for $i\neq
j$. Observe that, since $Q(\lambda)$ is an upper triangular matrix,
$\{v_i\}_{i=0}^\ell$ is a linearly independent set of vectors. Then if we set
$X=[v_{1},\ldots,v_\ell]$ and $\Lambda$ as the diagonal matrix whose $j$-th
diagonal element is given by  $\Lambda_{jj}=\lambda_j$, we have 
that the differential equation \eqref{hipergeometrica_esfericas} become equal to
$$z(1-z)F''(z)+(C-z(A+B+1))F'(z)-ABF(z)=0,$$
where $A=U-X\Lambda X^{-1}-1$ and $B=X\Lambda X^{-1}$.

\bibliography{database}

\end{document}